\documentclass[letterpaper, 10pt, conference]{ieeeconf} 

\IEEEoverridecommandlockouts

\usepackage{ucs}
\usepackage[utf8x]{inputenc}
\usepackage[english]{babel}
\usepackage{fontenc}
\usepackage{graphicx}
\usepackage{amsmath}
\usepackage{amsfonts}
\usepackage{amssymb}
\usepackage{algorithm}
\usepackage{algorithmic}
\usepackage[space]{cite} 
\usepackage{acronym}
\usepackage[breaklinks, bookmarksopen, bookmarksnumbered=true,pdftitle={Sajal Bhatia - Guarding Cyber-Physical Systems},pdfauthor={Sajal Bhatia}, colorlinks=true,linkcolor=black,citecolor=black, urlcolor=black]{hyperref}

\usepackage{breakurl} 
\usepackage{epsfig} 
\usepackage{todonotes} 
\usepackage{epstopdf}
\usepackage{subfigure}
\newtheorem{theorem}{\bf Theorem}[section]

\newtheorem{proposition}[theorem]{\bf Proposition}

\newtheorem{definition}{\bf Definition}[section]
\newtheorem{proof*}{\bf Proof}

\acrodef{CPS}{Cyber-Physical Systems}
\acrodef{SIS}{Susceptible Infected Susceptible}
\acrodef{SIR}{Susceptible Infected Recovered}
\acrodef{SIRS}{Susceptible Infected Recovered Susceptible}
\acrodef{hSIS}{Heterogeneous SIS}
\acrodef{hSIRS}{Heterogeneous SIRS}
\acrodef{CPS}{Cyber Physical Systems}
\acrodef{ICT}{Information and Communication Technology}
\acrodef{NLDS}{Non-Linear Dynamical System}

\title{\LARGE \bf Guarding Networks Through Heterogeneous Mobile Guards}

\author{Waseem Abbas, Sajal Bhatia and Xenofon Koutsoukos
\thanks{Waseem Abbas, Sajal Bhatia and Xenofon Koutsoukos are with Institute for Software Integrated Systems, Department of Electrical Engineering and Computer Science, Vanderbilt University, Nashville, TN 37212, USA
        {\tt\small \{waseem.abbas, sajal.bhatia, xenofon.koutsoukos\}@vanderbilt.edu}}%
        }%

\begin{document}

\maketitle
\thispagestyle{empty}
\pagestyle{empty}

\begin{abstract}
In this article, the issue of guarding multi-agent systems against a sequence of intruder attacks through mobile heterogeneous guards (guards with different ranges) is discussed. The article makes use of graph theoretic abstractions of such systems in which agents are the nodes of a graph and edges represent interconnections between agents. Guards represent specialized mobile agents on specific nodes with capabilities to successfully detect and respond to an attack within their guarding range. Using this abstraction, the article addresses the problem in the context of \textit{eternal security} problem in graphs. Eternal security refers to securing all the nodes in a graph against an infinite sequence of intruder attacks by a certain minimum number of guards. This paper makes use of heterogeneous guards and addresses all the components of the eternal security problem including the number of guards, their deployment and movement strategies. In the proposed solution, a graph is decomposed into clusters and a guard with appropriate range is then assigned to each cluster. These guards ensure that all nodes within their corresponding cluster are being protected at all times, thereby achieving the eternal security in the graph.

\end{abstract}

\section{INTRODUCTION}
\acresetall 
Networked systems such as \ac{CPS} have become an indispensable part of the modern society. Their ubiquitous presence in critical infrastructures such as power, water, and transportation has also led to growing concerns regarding their security against intruder attacks. An anomalous behavior by an individual agent may propagate and potentially result in the failure of the entire system. This not only demands a continuous surveillance of the system, but also the design and implementation of efficient mitigation strategies to minimize the overall impact of attacks, thereby motivating this study of guarding such systems.

Multi-agent systems can often be abstracted and modeled using a graph structure in which \textit{nodes} represent agents and \textit{edges} represent interconnections between these agents. This abstraction provides a framework to understand and analyze various system properties in terms of the underlying graph. It also provides a plethora of tools from graph theory that can be employed for an in-depth study of various problems of these systems. A specific problem in this domain is the \textit{eternal security} in graphs. In this article, this notion of eternal security in graphs, introduced in~\cite{burger2004infinite} and later studied in~\cite{goddard2005eternal,klostermeyer2007eternal,klostermeyer2009eternal}, is discussed and extended. 



The concept of eternal security addresses the problem of securing all the nodes in a graph against an infinite sequence of intruder attacks by a certain minimum number of guards. An \textit{intruder attack} on a node refers to any malicious activity on that node, for instance compromising a sensor node in a~\ac{CPS} to send fake values of the parameter being monitored. A \textit{guard} is an agent placed on a specific node that can detect and respond to an intruder attack within its range by moving from one node to another along the edges of a graph. If these guards are placed on nodes such that every node in the graph lies within the range of at least one guard, the graph is said to be \textit{secured} against an intruder attack on any of its nodes. The location of these guards is referred to as \textit{secure configuration}. The movement of guards along the edges from one node to another, however, might disturb this secure configuration as some nodes might end up not being within any guard's range as illustrated in Fig.~\ref{fig:intro1}. The notion of eternal security deals with such situations and ensures that all the nodes are secured against an arbitrary sequence of attacks. The objective is to determine a number of guards of given ranges, deploy them within a graph and outline a movement strategy such that a secure configuration is ensured for all the nodes at all times.

\begin{figure*}[!htb]
\begin{center}
\includegraphics[scale=0.9]{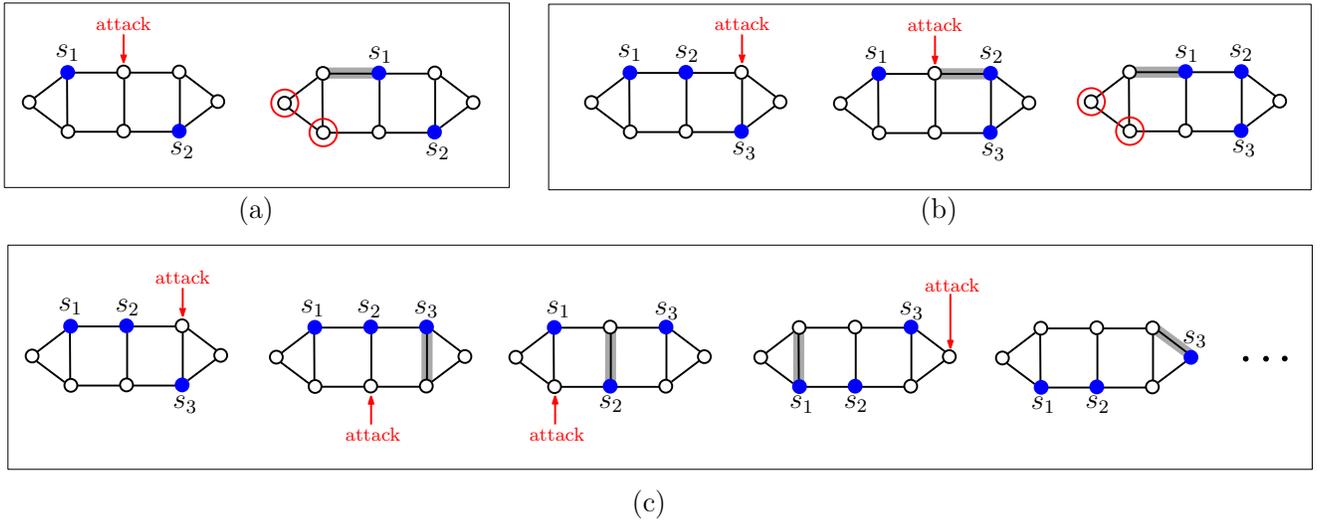}
\caption{(a) Two guards, $s_1$ and $s_2$, each capable of detecting and responding to an attack on
an adjacent vertex are securing the vertices of a graph through an initial secure configuration.
In the case of an attack on a vertex indicated by an arrow, $s_1$ moves towards it to
counter the attack. The resulting configuration of guards is unsecure as the circled vertices
have no guard in their neighborhoods. Here, the problem is that the \textit{number of guards} is not
sufficient. (b) Three guards $s_1$, $s_2$, and $s_3$, each having the range 1 are deployed. After
two intruder attacks, guards' configuration is unable to secure all the vertices. Though
the number of guards are sufficient in this case, the \textit{strategy} for the movement of guards to
counter attacks is not appropriate to eternally secure a graph against an arbitrary
sequence of attacks. (c) Guards move to counter attacks such that the resulting configuration
is always secure, i.e., for every vertex there always exists a guard to secure it. This
makes a graph eternally secure against any sequence of attacks.} 
\label{fig:intro1}
\end{center}
\end{figure*}

In many practical systems, 
these mobile guards can be thought of as unmanned devices (robots) placed at the gateway nodes, guarding the repeater and sensor nodes within their guarding ranges against malicious attacks. In scenarios like this, there may be a subset of nodes (gateways) that are more sensitive or critical than others (sensors) and therefore require more immediate consideration in case of an attack on any of these nodes. Eternally securing such a heterogeneous network requires making use of heterogeneous guards, i.e., guards with different ranges rather than having all guards with the same range. 

The problem of finding the number of guards required for the eternal security of graphs has been previously studied. Goddard et al.~\cite{goddard2005eternal} related this number to the domination number of a graph, whereas~\cite{goldwasser2008tight,klostermeyer2007eternal} provided bounds in terms of the independence number of a graph. In~\cite{klostermeyer2012vertex}, the required number of guards is compared to the vertex cover number. Recently, this problem is studied for the proper interval graphs in \cite{Braga1}, and a solution is provided in terms of the size of the largest independent set. In all these results, a guard was able to detect and respond to an attack only to the nodes in its immediate neighborhood, i.e., at 1 hop (edge length) distance. Abbas et al.~\cite{abbas2012securing} studied this problem for guards that can detect and respond to attacks that are at most $k$ hops from them. A limitation of the previous work done in this area has been that all guards were assumed to have same ranges, which may not be desirable in real-world scenarios discussed above. Besides this, all these studies have focused primarily on finding the number of guards required for the eternal security, paying less attention on their movement strategies, an essential component in achieving eternal security. 

This article studies the eternal security in graphs, addressing the aforementioned issues and in the process makes the following contributions:

\begin{itemize}
\item addresses the eternal security through a set of heterogeneous guards, i.e., guards with different ranges,
\item presents an algorithm for an appropriate placement and movement strategies of these heterogeneous guards to ensure eternal security.
\end{itemize}
 

%


\acresetall 

%



The remainder of the paper is organized as follow: Section~\ref{sec:preliminaries} introduces the terms that will be used throughout the paper. Section~\ref{sec:prob_formulation} formulates the problem addressed in this paper. Section~\ref{sec:proposed_solution} provides details of the proposed solution, and also presents an algorithm for decomposing a graph into clusters for the eternal security. Section~\ref{sec:example} illustrates the algorithm through an example and presents its evaluation. Finally, Section~\ref{sec:conclusion} summarizes the work presented in this paper.

\section{PRELIMINARIES}
\label{sec:preliminaries}
A graph $G(V,E)$, with a vertex set $V(G)$ and an edge set $E(G)$ is a simple, undirected graph. For simplicity, notations $V$ and $E$ are used for the vertex set and edge set of a graph respectively. An edge between vertices $u$ and $v$ is represented by $u\sim v$. Moreover, terms vertex and \textit{node} are used interchangeably. 
A \textit{complete} subgraph is induced by the vertex set $W\subseteq V$ whenever $u\sim v \in E$ for all $u,v \in W$. A subset of vertices inducing a complete subgraph is called a \textit{clique}. A clique that can not be extended by including one or more adjacent vertices is a \textit{maximal clique}. In other words, a clique not contained in any other clique of a larger size is a maximal clique. The determination of all maximal cliques in a graph is known as the \textit{maximal clique decomposition} of a graph. The distance between two vertices $u$ and $v$ in $G$, denoted by $d(u,v)_G$, is the length of the shortest path between $u$ and $v$ in $G$. Here, the \textit{path length} is the number of edges in a path. A path length of $r$ is sometimes referred to as the $r$-\textit{hop}. The \textit{diameter} of a graph, $diam(G)$, is $\max\; d(u,v)_G,\;\forall u,v\in G$. The $r^{th}$ \textit{power of a graph} $G$, denoted by $G^r$, is a graph with $V(G^r)=V(G)$ and $u\sim v \in E(G^r)$, whenever $d(u,v)_{G}\le r$.  

\begin{figure*}[!htb]
\begin{center}
\includegraphics[scale=0.65]{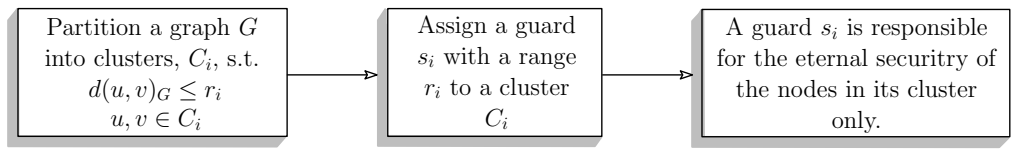}
\caption{A scheme for eternally securing a graph by making clusters.}
\label{fig:motivation}
\end{center}
\end{figure*}

\section{PROBLEM FORMULATION}
\label{sec:prob_formulation}
In this section, the problem of securing a network against a sequence of intruder attacks through a set of heterogeneous guards is formulated.

Let a network of agents interacting with each other be represented by a graph $G$, in which the vertex set $V$ represents agents, and the edge set $E$ corresponds to interactions between agents. Let $\mathcal{S}$ be a set of guards, in which each guard $s_i\in\mathcal{S}$ has some sensing and response range $r_i$ and is located on some vertex of $G$. A guard with the \textit{range}~$r_i$ can detect and respond to an intruder attack on a vertex that is at most $r_i$ hops away from it by marching towards the attacked vertex. The vertex at which a guard $s_i$ is present at time $k$ is described by the map $f$,
\begin{equation}
f: \; (\mathcal{S},k)\;\rightarrow\;V
\end{equation}

$f_k(s_i)$ is used to denote the vertex where $s_i$ guard is located at time instant $k$. Moreover, $f_k(\mathcal{S}) = \{f_k(s_i)\mid s_i\in \mathcal{S}\}$ is defined.

A \emph{vertex $v$ is secured} at time $k$ whenever it lies within a range of at least one guard at time instant $k$, i.e.,
\begin{equation}
\label{eqn:vsec}
\exists s_i:\; d(f_k(s_i),v)_G\;\le \;r_i
\end{equation}

If (\ref{eqn:vsec}) is true for all the vertices in a graph, then \textit{$G$ is secured} against an intruder at time $k$, and we say that $f_k(\mathcal{S})$ is a \textit{secure configuration} of guards in $\mathcal{S}$.

In the case of an intruder attack at some vertex $u\in V$ at time $k$, a guard $s_i$ securing $u$ will move from $f_k(s_i)$ to $u$ along the edges of a graph to counter an intruder. This results into a new configuration of guards, $f_{k+1}(\mathcal{S})$ at time $k+1$. If $f_{k+1}(\mathcal{S})$ is also a secure configuration, then the graph remains secured against an intruder attack.

\begin{definition}(Eternal Security)
A graph is eternally secure if for any $k$, a secure configuration of guards at time $k$ results into another secure configuration at time $k+1$ as a result of the movement of some guard $s_i\in\mathcal{S}$ from the vertex $f_k(s_i)$ to $f_{k+1}(\mathcal{S})$.
\end{definition}


Here it is assumed that at any time instant, there can be an attack only at a single vertex, and a single guard moves to counter this intruder. In other words, $\left | f_{k+1}(\mathcal{S})~-~f_k(\mathcal{S})\right |\le1$. Later, it is shown that under certain conditions, this assumption can be relaxed to $\left| f_{k+1}(\mathcal{S})-f_k(\mathcal{S})\right|\le\left |\mathcal{S}\right|$. The objective is to investigate the following problem,

\vspace{0.2cm}
\textit{How can a graph be eternally secured against a sequence of intruder attacks using a set of guards $\mathcal{S}$, where guards can have different ranges?}
\vspace{0.2cm}

Thus, the notion of eternal security in graphs has three major aspects; 

(a) Existence of a solution, i.e., if it is possible to eternally secure a graph through a given set of guards and their ranges. 

(b) Appropriate deployment of guards on various vertices in a graph. 

(c) Strategy for moving a right guard to counter an intruder attack such that the resulting configuration of guards is also secure.

\section{PROPOSED SOLUTION}
\label{sec:proposed_solution}
In this section, a scheme to distribute a given number guards with various ranges, among various vertices in a graph to make it eternally secure is presented. Let $\mathcal{S} = [s_1,s_2,\cdots,s_\sigma]$ be a vector of given guards, and $\textit{\textbf{r}} =[r_1,r_2,\cdots,r_\sigma]$ be a vector of their ranges, where $r_i$ is the range of a guard $s_i$. The proposed solution starts with the following simple observation,

\noindent
\textbf{Observation 1:}
{A single guard with a range $r$ makes a graph $G$ with a diameter $d$ eternally secure, if and only if $r\ge d$.}

The above observation provides a systematic way of distributing guards in $\mathcal{S}$ with their corresponding ranges in $\textit{\textbf{r}}$ to make a graph eternally secure. 

The proposed approach partitions a graph into clusters, and assigns a guard to each cluster which is then responsible for securing the nodes in its cluster only. These clusters are formed such that the distance between any two nodes of the same cluster is not greater than the range of the guard assigned to that cluster, i.e., $d(u,v)_G\le r_i$, where $u$ and $v$ are the vertices of the same cluster $C_i$, and $r_i$ is the range of the guard $s_i$ assigned to $C_i$. Since the distance between any two nodes in $C_i$ is not greater than $r_i$, guard $s_i$ sufficient for the eternal security of all the nodes in $C_i$. A block diagram of the scheme is shown in Fig.~\ref{fig:motivation}.


As an example, consider a graph shown in Fig.~\ref{fig:clusters}. Let there be three guards, $s_1,s_2$, and $s_3$, with ranges $1,2$ and $3$ respectively. The vertices of $G$ are partitioned into three clusters, and guards $s_1,s_2$ and $s_3$ are assigned to clusters $C_1,C_2$ and $C_3$ respectively. It is to be noted that for any cluster $C_i$, $d(u,v)_G\le r_i$, $\forall u,v\in C_i$. Therefore, a guard $s_i$ can always detect and respond to an intruder attack on some vertex in $C_i$ making $G$ eternally secure.

\begin{figure}[!htb]
\centerline{\epsfig{figure=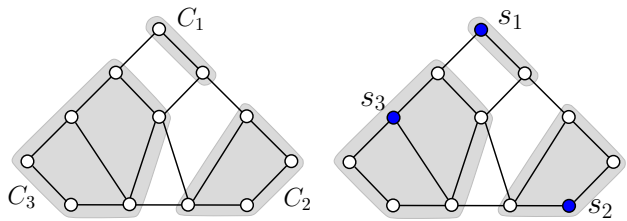,scale=0.69}}
\caption{\label{fig:clusters}An example illustrating the partitioning of graph vertices into clusters for eternal security. The guards $s_1,s_2,s_3$, with ranges $1,2$, and $3$ respectively are assigned to the clusters $C_1,C_2$, and $C_3$.}
\end{figure}

Under a secure configuration of guards within a graph, a node may be secured by more than one guard. In the case of an attack on that node, a response by one of the guards may result into another secure configuration, while a response by some other guard might produce a non-secure configuration of guards as shown in Fig.~\ref{fig:response}. Thus, for the eternal security, it is crucial to determine a right guard to be used to counter an attack. The clustering approach is particularly useful for that matter as guards will now respond to attacks on vertices in their clusters only. For a given number of guards and ranges, an effective partitioning of graph vertices into clusters is now discussed.

\begin{figure*}[!htb]
\centering
\subfigure[]{%
\includegraphics[scale = 1]{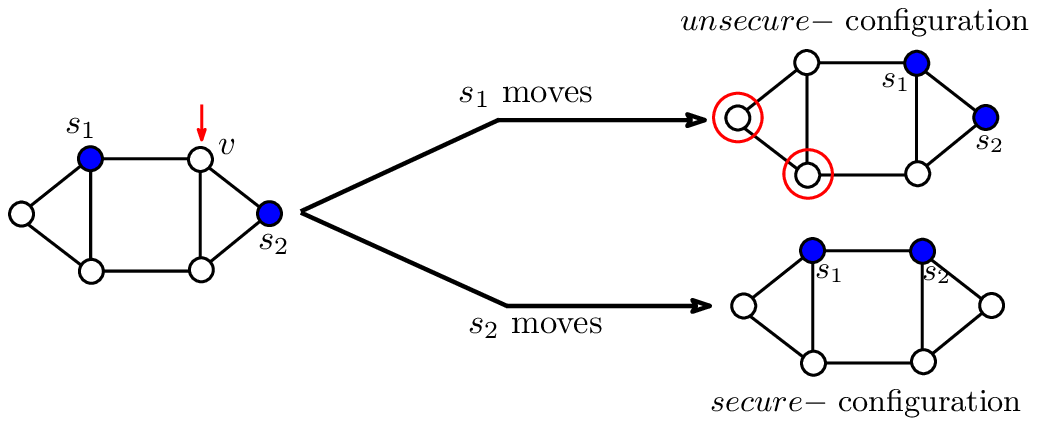}
\label{fig:response_a}}
\quad
\quad
\subfigure[]{%
\includegraphics[scale = 0.75]{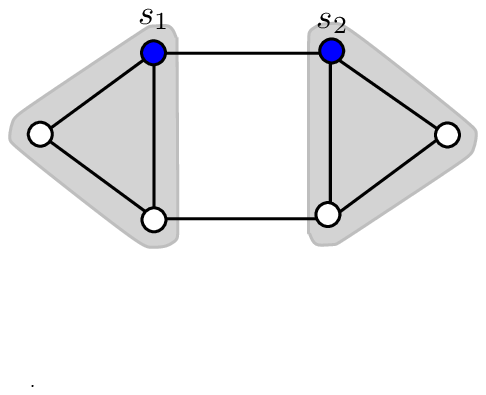}
\label{fig:response_b}}
\caption{(a) A secure configuration of guards $s_1$ and $s_2$, each having a range $1$ is shown. In the case of an attack at the vertex $v$, both guards can counter the attack since $d(s_1,v)=d(s_2,v)=1$. However, movement of the guard $s_1$ to $v$ results into an un-secure configuration as the circled nodes are not secured by any guard. On the other hand, movement of $s_2$ to counter the attack produces a secure configuration. (b) The vertices are partitioned into two clusters, each having a single guard. Each guard is responsible for the security of the vertices in its cluster only.}
\label{fig:response}
\end{figure*}

\subsection{Decomposition of a Graph into Clusters}
The objective is to obtain clusters of maximal sizes and assign guards with specific ranges to these clusters to eternally secure all the nodes. A guard is assigned to a cluster in such a way that the pair-wise distance between any two nodes in the cluster is at most equal to the range of the guard in that cluster. For a given graph and a set of guards along with their ranges, the aim is to decompose a graph into clusters such that the clusters include (cover) the maximum number of nodes in the graph.

The notion of graph power can be useful for the decomposition of a graph into clusters for the eternal security purpose. A guard with a range $r$ can eternally secure the nodes in $G$ that are within a distance of $r$ hops from each other. In other words, a guard with a range $r$ can
eternally secure all the vertices that induce a complete subgraph in the $r^{th}$ power of $G$ (i.e., $G^{r}$). Thus, for a guard with a range $r$, a cluster with the maximum number of vertices can
be obtained by maximal clique decomposition of $G^r$, and then selecting a maximum
clique. All the vertices in the maximum clique of $G^r$ can then be eternally secured
by a guard with a range $r$. For guards with various ranges (i.e., $r_i$), corresponding
clusters can be obtained by repeating the same procedure, i.e. by computing a maximal clique decomposition of $G^{r_i}$ and then selecting a clique with the
maximum number of vertices that are not yet secured. This process is formally defined in Algorithm~\ref{algo} and discussed in detail in the following section. 


\subsection{Main Algorithm}
This section presents an algorithm for partitioning a graph into clusters for the eternal security of a graph through a given set of guards with their respective ranges.

Let $\boldsymbol{\alpha}$ be the vector containing the ranges of guards. The $i^{th}$ element of $\boldsymbol{\alpha}$ is denoted by $\alpha_i$. Moreover, $\boldsymbol{\beta}$ be the vector in which the $i^{th}$ element, denoted by $\beta_i$, represents the number of guards with the range $\alpha_i$. For example, $\boldsymbol{\alpha}=[4,\; 2,\; 1]$, and $\boldsymbol{\beta}=[1,\; 3,\; 2]$ indicates that a graph has a single guard with a range $4$, three guards with a range $2$, and two guards have a range $1$. Furthermore, let $V_{uncov}$ be the set of vertices that are not included in any cluster.

In the initial phase, for each $\alpha_i$, maximal clique decomposition of $G^{\alpha_i}$ is performed. Out of all the cliques obtained, clique containing the maximum number of uncovered elements, say $m$ is selected. Let $m$ belongs to the clique decomposition of $G^{\alpha_j}$. If there exists a guard with the range $\alpha_j$ that has not been assigned to any cluster, then a cluster consisting of the vertices in $m$ is formed, and a guard with a range $\alpha_j$ is assigned to the cluster. This procedure is repeated until all guards are assigned to clusters, or all vertices are covered.

\begin{algorithm}
\caption{Decomposing a graph into clusters for the eternal security}
\begin{algorithmic} 
\STATE \textbf{Input:} $G,\;\boldsymbol{\alpha,\; \beta}$
\STATE \textbf{Initial:} $g=1;\;c_i=1,\;\forall\;i\in\{1,\cdots\left|\boldsymbol{\alpha}\right|\};\;V_{uncov}=V$
\FOR {$i=1$ to $\left|\boldsymbol{\alpha}\right|$}
\STATE $$ M_i\leftarrow \texttt{MaxClique}(G^{{\alpha_i}}) $$
\ENDFOR
\WHILE {$g \le $ Total no. of guards}
\STATE $$\mathcal{M} \leftarrow \bigcup_{\substack{i=1\\c_i\le\beta_i}}^{\left|\boldsymbol{\alpha}\right|}M_i$$
\STATE Find $m\in\mathcal{M}$ such that $\left|m\cap V_{uncov}\right|$ is maximum.
\STATE Let $M_j$ be clique decomposition that contains $m$ and has $c_j\le\beta_j$.
\STATE Vertices in $m$ constitute a cluster $\mathcal{C}_g$, and a guard with a range $\alpha_j$ is assigned to the cluster.
\STATE $c_j \leftarrow c_j + 1$;\; $g \leftarrow g+1$
\STATE $V_{uncov} \leftarrow V_{uncov} - m$
\ENDWHILE
\end{algorithmic}
\label{algo}
\end{algorithm}

\begin{figure*}[!htb]
\centerline{\epsfig{figure=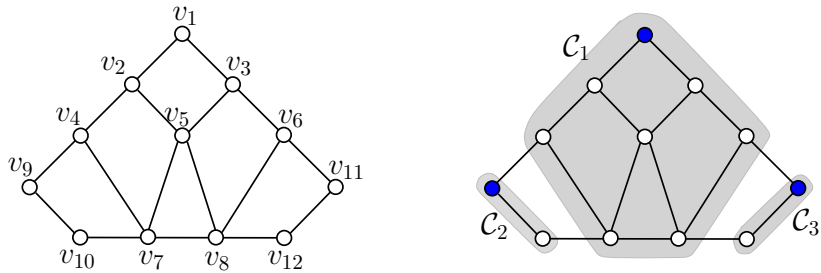,scale=0.8}}
\caption{\label{fig:example}A network with twelve nodes. $\mathcal{C}_1$ cluster contains the guard with the range $3$. Both $\mathcal{C}_1$ and $\mathcal{C}_2$ clusters have guards with the range 1.}
\end{figure*}

The algorithm uses maximal clique decomposition of a graph. Maximal clique decomposition is a well known combinatorial optimization problem. A number of theoretical and algorithmic results are available in the graph theory and computer science literature. Bron-Kerbosch algorithm \cite{bron1973algorithm} is a well known and widely used algorithm for finding maximal cliques in an undirected graph. Although other approaches have been reported, Bron-Kerbosch algorithm and its subsequent improvements are still regraded as one of the fastest and efficient ways to find maximal cliques \cite{stix2004finding}.

Furthermore, the problem of selecting cliques from the collection $\mathcal{M}$ is related to the \textit{maximum coverage problem}, in which the objective is to select a certain number of subsets from the collection to maximally cover the elements in the universal set. Maximum coverage problem is known to be NP-hard \cite{feige1998threshold}. In Algorithm 1, greedy approach is used for the selection of cliques. The greedy algorithm for the maximum coverage problem has an approximation ratio of $(1-1/e)$, which is essentially the best possible \cite{feige1998threshold}. Therefore, for the clustering problem, if \emph{Op} is the number of vertices that are included in some cluster by the optimal clustering algorithm, then

\begin{proposition}
For a given set of guards along with their ranges, Algorithm 1 includes at least $(1-1/e).Op$ number of vertices in some cluster. 
\end{proposition}


The following section gives a detailed example of the proposed algorithm to decompose a graph into clusters for the eternal security. It also evaluates the proposed algorithm in terms of the average distance moved by a guard in order to ensure eternal security of the graph.

\section{Example}
\label{sec:example}
As an example consider a network in Fig.~\ref{fig:example}. Let there be three guards $s_1,s_2$ and $s_3$ with ranges $1,1$ and $3$ respectively. These guards are to be distributed among the nodes to make the network eternally secure. Thus, the objective is to find an initial secure configuration of guards, and to specify a strategy to make sure that only the right guard moves to counter an attack on some node. Both of these goals can be achieved by decomposing a graph into clusters, and assigning an appropriate guard to each of these clusters. Using algorithm~\ref{algo}, the procedure starts by arranging guards' ranges in an array, $\alpha = [3\;,1]$. An array containing the number of guards corresponding to each of these ranges is $\beta = [1, \; 2]$. Since, $\alpha_1 = 3$, $G^3$ is computed. Maximal clique decomposition of $G^3$ using Bron-Kerbosch \cite{bron1973algorithm} algorithm gives,
\begin{equation*}
\begin{split}
 M_1 = & \{v_1,\cdots,v_8\},\{v_1, v_3,\cdots,v_6,v_8,v_{11},v_{12}\},\\
& \{v_1,\cdots,v_7, v_{10}\},\{v_1,v_2,v_4,\cdots,v_7,v_9,v_{10}\},\\
& \{v_1,\cdots,v_6,v_8,v_{12}\}.
\end{split}
\end{equation*}

Similarly, maximal clique decomposition of $G$ gives,
\begin{equation*}
\begin{split}
 M_2 = & \{v_1,v_2\},\{v_2, v_4\},\{v_2,v_7\},\{v_2,v_{10}\},\{v_4,v_5,v_6\},\\
& \{v_1,v_3\},\{v_3,v_4\},\{v_3,v_8\}, \{v_3,v_{12}\}, \{v_{11},v_{12}\},\\
& \{v_5,v_7\}, \{v_6,v_8\},\{v_8,v_{11}\}, \{v_7,v_9\}, \{v_9,v_{10}\}.\\
\end{split}
\end{equation*}

The counters $c_1$ and $c_2$ corresponding to $M_1$ and $M_2$ respectively are initialized to $1$. Moreover, $V_{uncov}=V$ initially. In the first iteration (of the while loop), $\mathcal{M} = M_1\cup M_2$. Since $\{v_1,\cdots,v_8\}\in M_1$ covers the most uncovered elements\footnote{If two sets cover the same number of uncovered vertices, a set is selected at random.}, the cluster $\mathcal{C}_1$ consisting of the vertices $\{v_1,\cdots,v_8\}$ is formed and the guard with the range $3$ is assigned to the cluster. The counter $c_1$ is incremented to $2$, and $V_{uncov}$ is set to $\{v_{9},\cdots,v_{12}\}$. In the next iteration, since $c_1>\beta_1$, $\mathcal{M}$ contains only $M_2$. It represents the fact that no more clusters can be formed that require guards with the range $3$. Thus, $\{v_9,v_{10}\}\in M_2$ covering the most vertices in $V_{uncov}$ is selected for the cluster $\mathcal{C}_2$. A guard with a range $1$ is then assigned to $\mathcal{C}_2$. Similarly, $\mathcal{C}_3$ comprising of $\{v_{11},v_{12}\}$ is formed and the remaining guard with the range $1$ is assigned to it. This cluster decomposition is illustrated in Fig.~\ref{fig:example}.

\section{Average Distance Moved by a Guard}
For the eternal security of a graph, guards with various sensing and response ranges are located on the vertices of a graph. In the case of an intruder attack on some vertex, they move along the edges of a graph through a path of vertices, thus covering a certain path length. The vertices of a graph are divided into clusters and all the vertices in a cluster are secured by a single guard with a range at least the maximum distance between any two nodes in the cluster. Since the maximum distance between any two nodes varies from cluster to cluster, the path lengths covered by guards to counter an attack also vary. The average distance a guard moves to eternally secure a graph may become a significant design parameter for various applications. Thus, analysis of the average distance moved by a guard for the eternal security of a graph by a cluster decomposition is also provided. Here, we assume that the probability of a vertex being attacked by an intruder is same for all the vertices (i.e., uniform probability distribution).

\begin{proposition}
Let $G$ be a graph whose vertices are decomposed into $\ell$ clusters $C_1, C_2,\cdots, C_{\ell}$ such that $\bigcup\limits_{i=1}^{\ell} C_i = V(G)$. For every cluster $C_i$, let there be a guard $s_i$ eternally securing all of the vertices in $C_i$ only. Then, the average distance (path length) moved by a guard to counter an intruder attack on some vertex $v\in V(G)$, denoted by $\tau$, is
\begin{equation}
\label{avgD}
\tau = \frac{1}{n}\left[\sum\limits_{i=1}^{\ell}\frac{1}{(n_i-1)}\sum\limits_{u,v\in C_i}d_G(u,v)\right]
\end{equation}
where $n_i$ is the number of vertices in the cluster $C_i$, and $n$ is the total number of vertices in $G$.
\end{proposition}

\begin{proof}
Let $v\in C_i$, then the average distance between $v$ and some other $u\in C_i$ in $G$ is defined as,
\begin{equation*}
\rho(v) = \frac{1}{n_i - 1} \sum\limits_{u\in C_i} d(u,v)_G
\end{equation*}
The average distance between the vertices in $C_i$, denoted by $\rho(C_i)$, is the average value of the distances between all pairs of vertices in $C_i$, i.e.,

\begin{equation*}
\begin{split}
\rho(C_i) = &\frac{1}{n_i}\sum\limits_{v\in C_i} \rho(v)
=\frac{1}{n_i(n_i - 1)}\sum\limits_{u,v\in C_i} d(u,v)_G
\end{split}
\end{equation*}

This is the average distance a guard $s_i$ moves in a cluster $C_i$ to counter an intruder attack on some $v\in C_i$. Since there are $\ell$ clusters with various number of vertices and guards with various ranges, the average distance a guard moves in response to an attack is the weighted average of $\rho(C_i)$.
\begin{equation}
\begin{split}
\tau & = \frac{1}{n} \sum\limits_{i=1}^{\ell} n_i \rho(C_i)\\
 &  = \frac{1}{n} \left[\sum\limits_{i=1}^{\ell} \frac{1}{(n_i -1)} \sum\limits_{u,v\in C_i} d(u,v)_G\right]
 \end{split}
 \end{equation} \end{proof}

For the network in Fig.~\ref{fig:example}, the average path length covered by a guard to counter an intruder attack as computed by (\ref{avgD}) is 1.523.

\noindent
\textit{Remarks:} In the proposed cluster decomposition, for any two vertices $u,v\in \mathcal{C}_i$, $d(u,v)_G\le r_i$, where $r_i$ is the range of the guard assigned to the cluster $\mathcal{C}_i$. Note that the distance considered here is the distance between the vertices in $G$, $d(u,v)_G$, and not in the subgraph induced by the vertices in the cluster $\mathcal{C}_i$, $d(u,v)_{\mathcal{C}_i}$. Since $d(u,v)_G\le d(u,v)_{\mathcal{C}_i}$, a guard with a range less than the diameter of the subgraph induced by the vertices in $\mathcal{C}_i$ may be sufficient for the eternal security of all the nodes in $\mathcal{C}_i$ as shown in the example of Fig.~\ref{fig:diam}. This makes our approach better than the one where a graph is simply decomposed into induced subgraphs with the diameters given by the guards' ranges.

\begin{figure}[!htb]
\centerline{\epsfig{figure=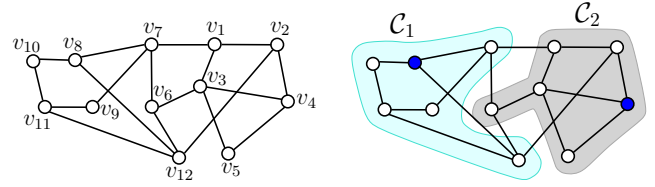,scale=0.74}}
\caption{\label{fig:diam}A network with twelve nodes. Two guards, each having a range $2$, are available. In (b) the cluster decomposition using Algorithm I is shown. Note that although $d(v_2,v_6)_{\mathcal{C}_2}=3$, both $v_2$ and $v_6$ are included in the same cluster $\mathcal{C}_2$ having a guard with a range $2$. It is possible as $d(v_2,v_6)_G = 2$. Also, there is no way to divide the given graph into two induced subgraphs each having a diameter at most 2. However, it is possible to have two clusters such that the distance between any two nodes of the same cluster is at most 2 as shown in the second figure.}
\end{figure}

\section{CONCLUSIONS}
\label{sec:conclusion}
In this article, the issue of guarding multi-agent systems against an infinite sequence of intruder attacks has been studied. Using graph theoretic abstractions, the problem of guarding networks against an infinite sequence of intruder attacks is formulated as the eternal security problem in graphs. Moreover, a way to achieve eternal security in graphs through heterogeneous guards is presented. For a given set of guards along with their respective ranges, a graph can be decomposed into clusters and an appropriate guard can be assigned to each cluster to achieve eternal security. The issues of guards' deployment and their movement strategies have been addressed using the proposed cluster decomposition approach. 
It is to be noted that in the case of a multi-attack situation where each cluster is attacked by at most one intruder, the proposed solution still ensures eternal security. However, in the case of multiple attacks within the same cluster, secure configuration of guards might not be maintained, and hence, eternal security might not be achieved. As a part of the future work, we want to extend this work to incorporate multiple attacks, even within the same cluster, while achieving eternal security. 



%

\section{Acknowledgments}
This work is supported in part by the National Science Foundation (CNS-1238959, CNS-1035655), and the Air Force Research Laboratory under Award FA8750-14-2-0180.

\end{document}